\documentclass[acmsmall,nonacm]{acmart}

\AtBeginDocument{%
  }

\usepackage{graphicx}
\usepackage{enumitem}
\usepackage{amsmath}
\usepackage{amsthm}
\usepackage{tikz}
\usepackage{algorithmic}

\newtheorem{theorem}{Theorem}
\newtheorem{lemma}{Lemma}

\begin{document}

\title{Fast, Fair and Truthful Distributed Stable Matching for Common Preferences}

\author{Juho Hirvonen}
\authornote{Both authors contributed equally to this research.}
\email{juho.hirvonen@aalto.fi}
\affiliation{%
\institution{ Helsinki Institute for Information Technology HIIT and Aalto University}
%\institution{Helsinki Institute for Information Technology HIIT and Aalto University}
  %\streetaddress{PO BOX 11000}
  \city{Espoo}
  \country{Finland}
  %\postcode{FI-00076}
}

\author{Sara Ranjbaran}
\email{sara.ranjbaran@aalto.fi}

\affiliation{%
  \institution{ Aalto University}
  %\streetaddress{PO BOX 11000}
  \city{Espoo}
  \country{Finland}
  %\postcode{FI-00076}
}
\
\begin{abstract}
  Stable matching is a fundamental problem studied both in economics and computer science. The task is to find a matching between two sides of agents that have preferences over who they want to be matched with. A matching is stable if no pair of agents prefer each other over their current matches. The deferred acceptance algorithm of Gale and Shapley solves this problem in polynomial time. Further, it is a mechanism: the proposing side in the algorithm is always incentivised to report their preferences truthfully.
  
  The deferred acceptance algorithm has a natural interpretation as a distributed algorithm (and thus a distributed mechanism). However, the algorithm is slow in the worst case and it is known that the stable matching problem cannot be solved efficiently in the distributed setting. In this work we study a natural special case of the stable matching problem where all agents on one side share common preferences. We show that in this case the deferred acceptance algorithm does yield a fast and truthful distributed mechanism for finding a stable matching. We show how algorithms for sampling random colorings can be used to break ties fairly and extend the results to fractional stable matching.
\end{abstract}

\begin{CCSXML}
<ccs2012>
<concept>
<concept_id>10003752.10010070.10010099.10010101</concept_id>
<concept_desc>Theory of computation~Algorithmic mechanism design</concept_desc>
<concept_significance>500</concept_significance>
</concept>
<concept>
<concept_id>10003752.10003809.10010172</concept_id>
<concept_desc>Theory of computation~Distributed algorithms</concept_desc>
<concept_significance>500</concept_significance>
</concept>
</ccs2012>
\end{CCSXML}

\ccsdesc[500]{Theory of computation~Algorithmic mechanism design}
\ccsdesc[500]{Theory of computation~Distributed algorithms}

\keywords{stable matching, deferred acceptance algorithm, local algorithm, mechanism design}
\maketitle

\section{Introduction}
% !TEX root = sample-acmsmall-submission.tex

\begin{figure}[t]
	\centering
	\includegraphics[width=0.9\textwidth]{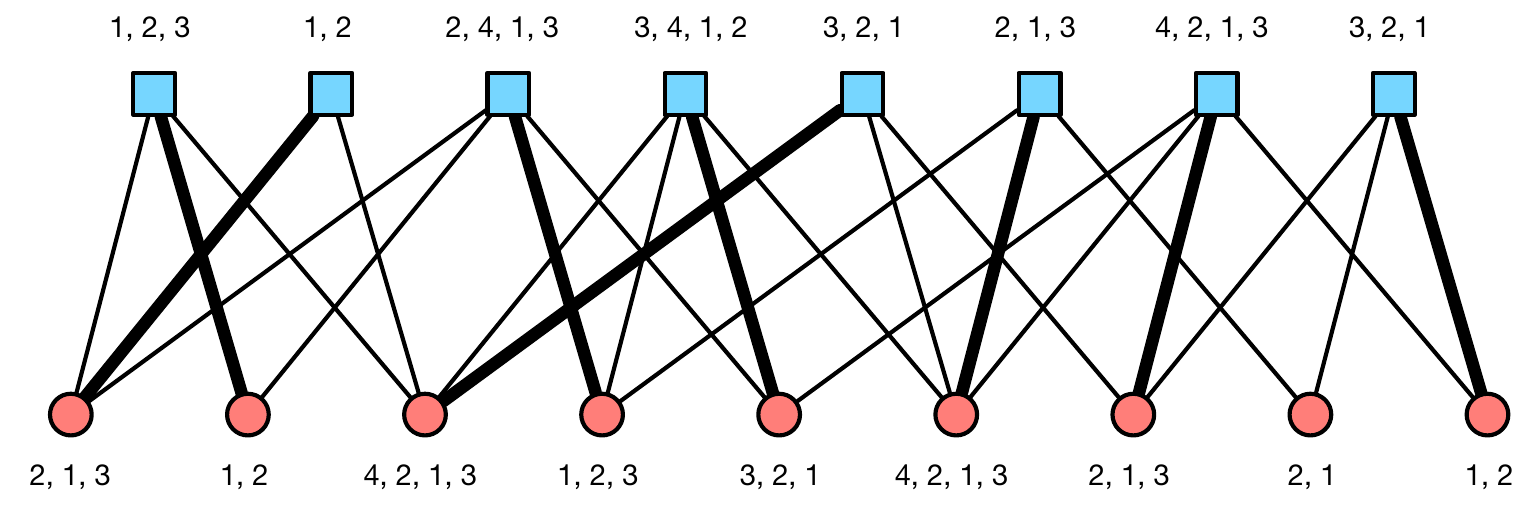}
	\caption{A matching instance and a stable matching in bold. The preferences are given as a list for the neighbours from left to right, with 1 indicating the highest preference. The matching is \emph{stable}: there is no pair of unmatched agents connected by an edge such that they would prefer each other over their current match. Since there are more agents on one side, some agents must remain unmatched.} \label{fig:sm}
\end{figure}

In mechanism design the goal is to design optimisation algorithms that work when the input is private information held by strategic agents. For example, in the case of auctions, each agent has a private valuation for the item being sold and the goal is find an auction mechanism that incentivises the agents to reveal their true valuations while assigning the item to the agent that values it the most.

Stable matching is one of the seminal problems studied in mechanism design~\cite{nobel12}. There is a bipartite matching graph where each node is an agent and each edge indicates a potential match. The agents have a linear order over their neighbours that is initially only known to them. Each agent reports \emph{some} linear order and based on these an algorithm computes a matching. The \emph{deferred acceptance} algorithm of Gale and Shapley~\cite{gale62sm} finds a stable matching: there is no pair of agents that would prefer each other over their assigned matches. See Figure~\ref{fig:sm} for an illustration.

The algorithm of Gale and Shapley consists of rounds of proposals. One side of the bipartition is the proposing side and one is the receiving side. In each round each unmatched proposing-side agent proposes to its next most preferred match. Agents that receive proposals accept the best proposal if they prefer it over their current match (deferred acceptance).

The deferred acceptance algorithm has the important property that, for the proposing side, revealing their true preferences is a dominant strategy~\cite{dubins81machiavelli,roth82economics}. This makes the deferred acceptance algorithm one of the few examples of a mechanism without money. It is also known that there is no mechanism that finds a stable matching and is truthful for all agents~\cite{roth82economics}.

The deferred acceptance algorithm can be applied to any setting where two sides need to matched based on preference criteria held by the individual agents. The original motivation of Gale and Shapley was college admission: how to match students to school seats~\cite{gale62sm}. The algorithm naturally generalises to this kind of one-to-many matching, where a single school can be matched to many students. Due to its simplicity and strategic properties, the deferred acceptance algorithm has also been widely deployed in practice~\cite{nobel12,roth2008deferred}, for example in the National Resident Matching Program (NRMP) in the US~\cite{roth84evolution} and New York City and Boston public high schools~\cite{roth08what}.

The deferred acceptance algorithm has a natural interpretation as a distributed proposal algorithm. However, its worst-case running time is $O(n^2)$ rounds~\cite{ostrovsky2015fast} and computing a stable matching requires $\Omega(n)$ rounds in the LOCAL model~\cite{floreen10stable}. We study a special case of the general stable matching problem where one side has common preferences, i.e.\ the preferences of all agents are restrictions of the same partial order~\cite{irving2008stable}. This special case corresponds to the case where the preferences of one side come from some objective ranking, such as standardised tests in school choice. This is a common assumption when studying matching between autonomous agents, for example in biology~\cite{ALPERN05}, economics~\cite{burdett97marriage}, game theory~\cite{eriksson2008instability}, and resource assignment in computer networks~\cite{GaoHY22}.

\subsection{Our contribution}

We present several extensions of the deferred acceptance algorithm in the setting where one of the sides has common preferences. Our algorithms implement the deferred acceptance algorithm as an efficient distributed mechanism~\cite{feigenbaum07distributed}, and thus can be applied in settings where a single centralised entity running the mechanism is infeasible. Our algorithms are simple and retain the potential of the original deferred acceptance algorithm for practical applications. 

We will next present informal versions of our main results. We show that for one-sided common preferences, the deferred acceptance algorithm can be implemented efficiently as a distributed algorithm. Let $S$ denote the number of different preference classes in the common preferences, and let $\Delta$ denote the product of the maximum degrees of the two sides in the matching graph. Let $n$ denote the number of agents on the proposing side.

\begin{theorem}[Informal, Theorem~\ref{thm:da-arbitrary-tb}]
	There is an incentive-compatible distributed implementation of the deferred acceptance algorithm in the CONGEST model  for one-sided common preferences, with running time $O(\Delta S + \log^* n)$.
\end{theorem}
The formal version of this Theorem is given in Theorem~\ref{thm:da-arbitrary-tb} in Section~\ref{sec:ff-da}.

Our algorithms are incentive-compatible, i.e.\ the agents on the proposing side are never worse off by reporting their true preferences. We believe that this is the first example of an efficient ($O(\operatorname{poly} \log n)$-time) distributed incentive-compatible mechanism.

We also show how to extend the algorithm to the fractional case that we call \emph{capacitated fractional stable matching}. The matching itself can take fractional values and both sides have a maximum capacity. This can be seen as a generalisation of stable matching to a setting where the clients have divisible jobs that are assigned to the providers for completion.

For general preferences, this setting is not well defined, as one must first define how agents preferences weight quality of match against the total amount of a match~\cite{CaragiannisFKV19}. We show that in the case of one-sided common preferences, this is not an issue and the natural greedy extension of the deferred acceptance algorithm is incentive-compatible as the truth maximises both both the quality and the quantity of the matching.

\begin{theorem}[Informal, Theorem~\ref{thm:fractional-det-da}]
	There is an incentive-compatible distributed implementation of a generalised deferred acceptance algorithm for capacitated fractional stable matching in the CONGEST model for one-sided common preferences, with running time $O(\Delta S + \log^* n)$.
\end{theorem}

We assume that the preferences may contain ties. In order to retain the common preferences, the tie-breaking procedure must be consistent: If each agent on the receiving side would break ties arbitrarily, in the worst case this tie-breaking could lead to a long propagation in the deferred acceptance algorithm. This issue is solved by using graph coloring as a tie-breaking procedure. While any coloring is incentive-compatible, we also want the coloring to be fair. We propose a simple and fast algorithm for sampling a random coloring and also show how existing algorithms for sampling random colorings can be used as a subroutine in our algorithm.

The first algorithm samples an almost uniformly random coloring using an existing subroutine~\cite{fischer18simple}. Let $\delta > 0$ denote the \emph{total variation distance} from uniform random coloring.
\begin{theorem}[Informal, Theorem~\ref{thm:fair-almost-uniform}]
		There is an incentive-compatible, randomised distributed implementation of the deferred acceptance algorithm in the CONGEST model for one-sided common preferences with a tie-breaking routine that has total variation distance $\delta > 0$ from random and running time $O(\Delta S + \log (n/\delta))$.
\end{theorem}
A formal version of this Theorem is given in Theorem~\ref{thm:fair-almost-uniform}, Section~\ref{ssec:sample-coloring}.

The second algorithm uses a simpler color sampling that only guarantees that some subset of agents succeed in being properly colored. These agents have a tie-breaking rule that is conditionally uniform, but the remaining agents have an arbitrary tie-breaking.
\begin{theorem}[Informal, Theorem~\ref{thm:fast-fair-da}]
	For any $\delta > 0$ there is an incentive-compatible, randomised distributed implementation of the deferred acceptance algorithm in the CONGEST model for one-sided common preferences with a tie-breaking routine that is uniformly random for a subset of $(1-\delta)n$ proposing-side agents and runs in time $O(\delta^{-1}\Delta S + \log^* n)$.
\end{theorem}
The formal version is given in Theorem~\ref{thm:fast-fair-da}, Section~\ref{ssec:fast-coloring-failures}.

Both of these fair tie-breaking procedures can also be applied to capacitated fractional stable matching, yielding algorithms with the same running times and guarantees.

Since finding a stable matching with general preferences is hard in the distributed setting~\cite{floreen10stable}, our results suggest a trade-off for market design: in certain scenarios, it might be preferable to force one of the sides to use common preferences in order to make the matching procedure more efficient.

\subsection{Related work}

The \emph{deferred acceptance} or \emph{Gale-Shapley algorithm} computes a stable matching: a matching (i.e.\ a subset of edges such that each agent belongs to at most one edge) such that no unmatched potential pair prefers each other over their assigned matches~\cite{gale62sm}. As a mechanism it is incentive-compatible for the proposing side of the algorithm~\cite{dubins81machiavelli,roth82economics}. There does not exist a mechanism (without payments) that is incentive-compatible for both sides at the same time~\cite{roth82economics}. We refer to a survey by Knuth~\cite{knuth1997stable} and the survey by Roth~\cite{roth2008deferred} for more details.

Irving, Manlove, and Scott studied the stable matching problem with common preferences (called master preference lists there)~\cite{irving2008stable}. This setting was also studied by Scott~\cite{scott05thesis} and O'Malley~\cite{omalley07thesis}. Generalisations of this setting were studied by Kamiyama~\cite{kamiyama15,kamiyama19}.

In the distributed setting, Khanchandani and Wattenhofer studied a variant of stable matching where preference lists are similar on one side~\cite{khanchandani16distributed}. The running time of their algorithm is $O(\Delta n)$, where $\Delta$ is the similarity parameter. Amira, Giladi, and Lotker studied a variant where the preferences are given as edge weights and gave an algorithm that finds a stable matching in $O(\sqrt{n})$ rounds~\cite{amira10distributed}. 

Flor{\'{e}}en, Kaski, Polishchuk, and Suomela~\cite{floreen10stable} showed that finding a stable matching requires $\Omega(n)$ rounds in the LOCAL model. Kipnis and Patt-Shamir showed that solving the stable matching problem, even in graphs of diamater $\Theta(n)$, requires $\Omega(\sqrt{n}/\log n)$ rounds in the CONGEST model. Most variants of stable matching are at least as hard as finding a maximal matching, and this is known to require $\Omega(\Delta)$ rounds, where $\Delta$ is the maximum degree of the graph~\cite{balliu21lower}.

Since finding a stable matching is hard in the distributed setting, there has been a lot of interest in finding almost stable matchings. An edge is unstable if it is not in the matching but both of its endpoints prefer each other to their current matches, and a matching $M$ is $\varepsilon$-stable if the number of unstable edges is at most $\varepsilon|M|$. Flor{\'{e}}en, Kaski, Polishchuk, and Suomela~\cite{floreen10stable} showed that truncating the deferred acceptance algorithm after $O(\Delta^2 / \varepsilon$) rounds results in an $\varepsilon$-stable matching. Ostrovsky and Rosenbaum present an $O(\log^5 n)$-round deterministic algorithm and an $O(\log^2 n)$-round randomised algorithm for computing almost stable matchings. Hassidim, Mansour, and Vardi \cite{hassidim2016local} study the problem in the centralised \emph{local computation algorithm (LCA)} model. It is important to note that none of these algorithms have the strategic guarantees of the original deferred acceptance algorithm, as we discuss in Section~\ref{ssec:da}.

We also study fractional stable matchings, which have been studied in the literature due to their connection to integral stable matchings~\cite{vandevate89linear,roth93stable}. Caragiannis, Filos{-}Ratsikas, Kanellopoulos, and Vaish \cite{CaragiannisFKV19} resolve the issue ill-defined utility by studying cardinal preferences. See their work for further references on the subject.

Our algorithms can be seen as examples of \emph{distributed algorithmic mechanism design (DAMD)}~\cite{feigenbaum07distributed}. In DAMD the goal is to design mechanisms that can be run by the participating agents. Work in this field has typically dealt with distributed mechanisms for global problems in distributed networks, such as multi-cast cost sharing~\cite{feigenbaum01sharing}, leader election~\cite{AbrahamDH19}, consensus, and renaming~\cite{AfekGFS13}. We are not aware of any distributed mechanisms implemented by a local algorithm before our work.

\vspace{-\baselineskip} % Adjust the space
%\vspace*{4pt}
\vspace*{10pt}
\section{Background and definitions} \label{sec:def}
% !TEX root = sample-acmsmall-submission.tex

\subsection{Distributed algorithms}

We study distributed algorithms in the standard LOCAL and CONGEST models of distributed computing~\cite{Peleg2000,Linial1992}. The system is represented as a graph $G = (V,E)$ where the set of nodes $V$ are the computational agents and the edges $E$ are the communication links: two nodes $u$ and $v$ can communicate directly if and only if $\{u,v\} \in E$. Nodes connected to $u$ are its neighbors and are denoted by $N(u)$. The degree (number of neighbours) of node $v$ is denoted by $\deg(v)$. The maximum degree $\Delta$ of a graph $G$ is $\max_{v \in V} \deg(v)$. We assume that $n$ and $\Delta$ are known to all nodes. Additionally, as input, each node has a unique identifier from $[\operatorname{poly}(n)]$. 

Computation proceeds in synchronous rounds. In each round, each node can send a message to each of its neighbors, receive the messages sent by its neighbors, and update its state. In the LOCAL model these messages can be unbounded in size, and the CONGEST model only differs from the LOCAL in that the size of the messages is $O(\log n)$ bits. Eventually each node must stop and announce its own output. The running time of an algorithm is the number of communication rounds until all nodes have stopped.

Unless otherwise stated, all algorithms are deterministic. In randomised algorithms, nodes additionally have a uniformly random and sufficiently long string of bits as input.

\subsection{Mechanism design}

A mechanism $M$ is a game: It consists of a set $V$ of \emph{agents} and each agent $v \in V$ has a private \emph{type} $t(v)$ representing its utility function. In the case of stable matching, the type of the agent is its preference relation over its neighbours. For each agent there is a set of possible \emph{actions} $A(v)$ representing the possible ways of reporting its type (preferences) to the mechanism. The \emph{strategy} $a(v)$ of agent $v$ is its chosen action. The mechanism has an \emph{allocation function} $F$ that maps the strategies of the agents to a set $O$ of possible \emph{outcomes}.

A mechanism is \emph{incentive-compatible} or \emph{truthful} if reporting the truth (i.e.\ choosing $a(v) = t(v)$) is a dominant strategy: agent is never worse off by reporting truthfully. For randomised mechanisms, we say that a mechanism is (universally) incentive-compatible if truth is a dominant strategy for all choices of randomness.

In this work we consider \emph{distributed mechanisms} where the allocation function can be implemented as a distributed algorithm. Each agent corresponds to a node in graph $G$. The strategy of each agent is part of the input of the corresponding node. After the strategies are revealed in the input, the agents can no longer affect the mechanism and we assume that the algorithm implementing the allocation rule is run.

\subsection{Stable matching and deferred acceptance} \label{ssec:da}

We study the graph-theoretic version of the stable matching problem. The agents form a bipartite matching graph $G = (C \cup P, E)$. We call the side $C$ \emph{clients} and the side $P$ \emph{providers}: the clients correspond to the proposing side in the deferred acceptance algorithm. The edges represent the possible matches of each agent, and these are all assumed to be preferable to being unmatched. We assume that the communication network $N = (C \cup P, E')$ is $G$ with additional edges between any pair of clients that are connected to the same provider. The modelling choice is somewhat arbitrary, and for most parameter values it allows us to simplify the description of the CONGEST algorithms. For some large values of $\Delta$, the maximum degree of the network, there is a difference between the running time of the best LOCAL and CONGEST algorithms if the matching graph is the communication network. We discuss this in Section~\ref{ssec:congest-col}.

When we refer to degrees, we mean the degrees of the matching graph unless otherwise specified. We denote by $\Delta_C$ the maximum degree among the clients and by $\Delta_P$ the maximum degree among the providers, and assume that these are known to the agents.

The type of each client and each provider is a linear order on its incident providers, determining in which order it prefers to be matched with them. The basic deferred acceptance algorithm can be implemented as a distributed algorithm as follows. Let $>_v$ denote the announced preferences of agent $v$. Each client $v$ initialises a list $L(v)$ of its $\deg(v)$ neighbours in the preference order, starting from the most preferred. Each client and provider $v$ maintains a match variable $m(v)$ initially set to unmatched.

On odd rounds $2i-1$ for $i=1,\dots$:
\begin{enumerate}[noitemsep]
  \item Each unmatched client $v$ sends message \emph{proposal} to its most preferred match in $L(v)$ and removes it from $L(v)$.
  \item Each matched client waits.
  \item Each provider $p$ that receives one or more proposals picks the one that is from the most preferred client $v$, and checks if $v >_p m(p)$. If yes, set $m(p) \leftarrow v$.
\end{enumerate}

On even rounds $2i$ for $i=1,\dots$:
\begin{enumerate}[noitemsep]
  \item Each provider $p$ sends $m(p)$ to its neighbours (as acknowledgement of proposal acceptance or rejection).
  \item Each client $v$ that receives $m(p) = v$ from a neighbour sets $m(v) \leftarrow p$.
\end{enumerate}

The algorithm runs in $O(n^2)$ rounds. It is known that computing a stable matching requires $\Omega(n)$ rounds~\cite{floreen10stable}.

\begin{theorem}[Floreen et al.~\cite{floreen10stable}] \label{thm:linear-lb}
  Computing a stable matching in a system where each client and provider has maximum degree 2 requires $\Omega(n)$ rounds in the LOCAL model.
\end{theorem}

\begin{figure}[ht]
  \centering
  \includegraphics[width=0.7\textwidth]{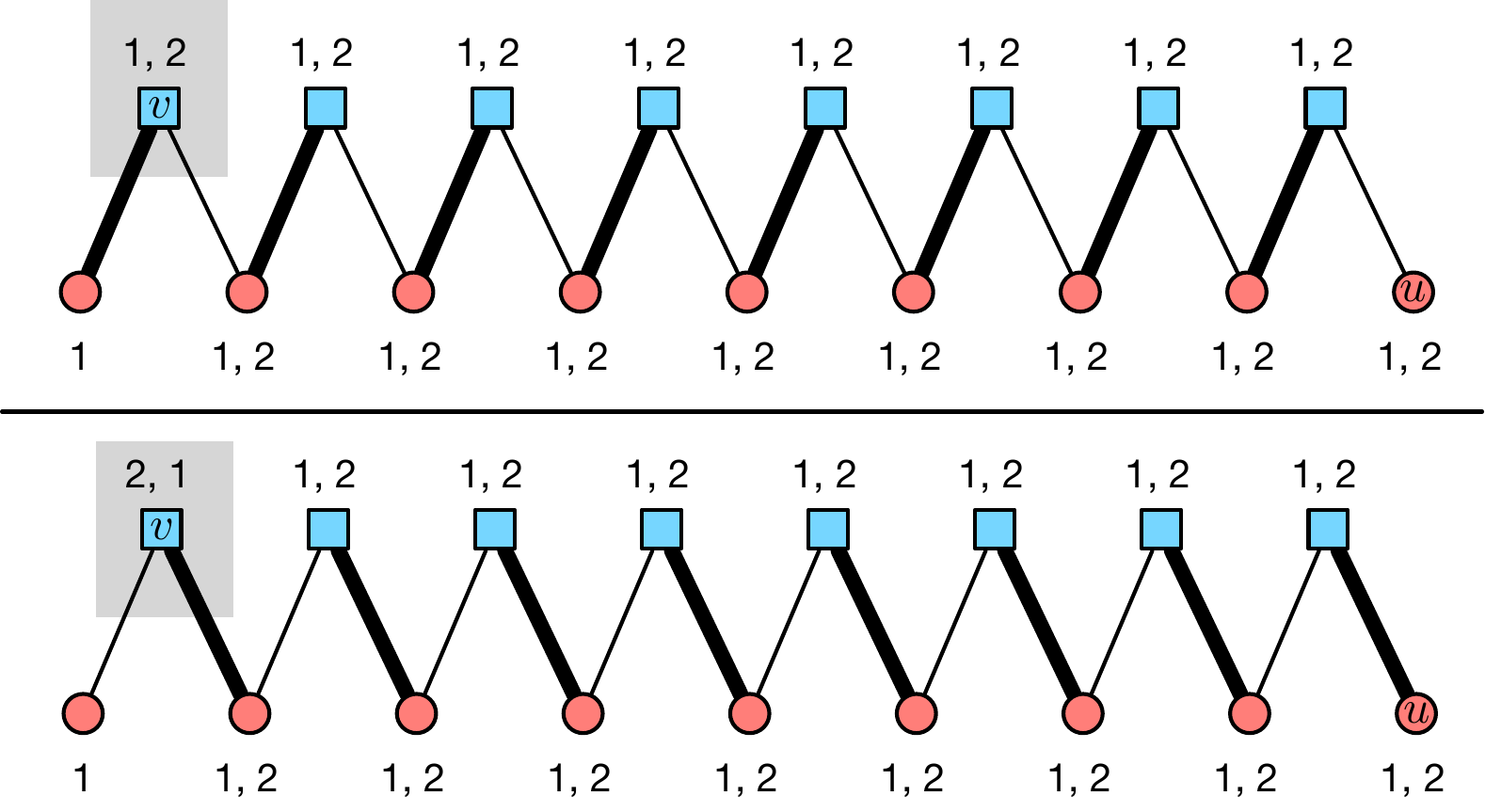}
  \caption{Two instances of stable matching that only differ in the preferences of agent $v$ on the left. This causes the unique stable matching to change, and therefore agent $u$ on the right has to learn the preferences of $v$ in order to compute a stable matching.} \label{fig:lb}
\end{figure}

The theorem follows from a simple argument: for any path, there exist two assignments of preferences to the agents such that 1) they only differ at one agent at one end of the path and 2) the resulting instances both have a unique stable matching and these matchings are edge-disjoint. See Figure~\ref{fig:lb} for an illustration.

To circumvent this hardness result, we study a special case of the stable matching problem where the providers have \emph{common preferences}: they always agree when comparing two clients. We assume that the preferences of the providers can be expressed as a single \emph{score function} $s \colon C \rightarrow \{1,\dots,S\}$. We call the value $s(v)$ of agent $v$ its \emph{preference class}. 

If the preferences are represented in a more fine-grained manner, for example as real numbers from an interval~\cite{burdett97marriage}, then by discretising them into $S$ intervals it is possible to achieve a trade-off between solution quality and running time. While explicitly knowing the possible range of values of $s$ is not important, it is important to know an upper bound on the number of distinct preference classes imposed by $s$ in order for the algorithm to terminate.

\paragraph{Variants of deferred acceptance.} It is typically assumed that the preferences of the agents are strict. However, the deferred acceptance algorithm generalises naturally to the case where there are ties by using a suitable tie-breaking mechanism. It is important to note that, in order to maintain incentive-compatibility, the procedure to break ties should not depend on the preferences of the agents. For example, a tie-breaking rule preferring the existing match in case of a tie could incentive agents to lie in order to obtain a match early.

The deferred acceptance algorithm also generalises to one-sided one-to-many matching. For example, in the case of school choice, schools have multiple seats and can thus accept multiple applicants. A simple reduction represents a provider with $k$ seats as $k$ identical copies of a provider with unit capacity. It is also possible to modify the algorithm itself so that the providers accept $k$ best matches. This has the advantage of not increasing the maximum degree of the graph.

However, the deferred acceptance algorithm does not generalise to the case of many-to-many matching, for example clients having multiple jobs that need to be assigned to providers for completion. In this case it is not clear what the utility function of a client is: for example, how does it prefer match quality compared to match size? We show in Section~\ref{sec:generalisation} that in the case of common preferences, there is a natural generalisation that can be solved efficiently.

\paragraph{Almost stable matching.} Several existing algorithms solve the almost stable matching problem~\cite{floreen10stable, ostrovsky2015fast} in the distributed setting. An edge is unstable if it is not in the matching but both of its endpoints prefer each other to their current matches. The goal is to find a matching $M$ with at most $\varepsilon|M|$ unstable edges.

It is important to note that previous efficient algorithms for almost stable matching are not incentive-compatible. For example, truncating the deferred acceptance algorithm leads to an almost stable matching~\cite{floreen10stable}. This algorithm, however, is not incentive-compatible. As an example, it could be that in the last round, agent $v$ becomes unmatched from $u$ when it accepts a better proposal. Since this was the last round, $v$ will remain unmatched. Now assume it had a neighbor $u'$ that also is unmatched: by lying that $u' >_v u$ agent $v$ could have ensured that it is matched with $u'$.

\section{Fast and fair deferred acceptance algorithm} \label{sec:ff-da}

In this section, we show that the deferred acceptance algorithm has a fast distributed implementation when one side has common preferences.

\subsection{Our basic algorithm} \label{ssec:basic-alg}

Since we assume that the clients are scored with a bounded number of values, we may have ties among the clients connected to any provider. Since the original deferred acceptance algorithm assumes that the preferences are distinct, these ties must be broken in a way that preserves the truthfulness of the mechanism and does not create long dependencies that would slow down the algorithm. We do this by computing a coloring of the virtual conflict network formed by clients that share a preference class and a provider. These colors are used to break ties: smaller color is preferred over a larger color if the preference classes are equal. Since the tie breaking rule is precomputed, clients cannot affect it by lying about their preferences, and the mechanism remains truthful. We note that tie-breaking can affect the quality of the solution, in particular when the preference lists are not complete. We do not attempt to optimise this.

The algorithm consists of three phases:
\begin{enumerate}[noitemsep, label=\arabic*)]
  \item Given the matching graph $G = (C \cup P, E)$, construct the virtual conflict network $H = (C, E')$, with the set of clients $C$ as nodes, where two clients $v$ and $u$ are connected by an edge $\{v,u\} \in E'$ if and only if $s(v) = s(u)$ and there exists $w \in P$ such that $\{v,w\}, \{u,w\} \in E$. \label{step:form-conflict}
  \item Compute a $c$-coloring $\phi$ of $H$ for some $c$. Create a new score function $s \circ \phi$ defined as $(s \circ \phi)(v) = c(s(v)-1)+\phi(v)$. \label{step:coloring}
  \item Run the distributed deferred acceptance algorithm with the provider-side common preferences given by the score function $s \circ \phi$. \label{step:da}
\end{enumerate} 

In step~\ref{step:coloring}, any coloring that does not depend on the preferences of the clients preserves truthfulness. In Section~\ref{sec:fairness} we discuss how the coloring can be done in a \emph{fair} way. Without specifying the coloring algorithm, we get the following Lemma.

\begin{lemma} \label{lem:mechanism-main}
  There exists a distributed mechanism in the CONGEST model for stable matching when the providers have common preferences. The mechanism is incentive-compatible for the clients and the algorithm implementing it runs in $O(T_c(n,\Delta_H)+Sc)$ rounds, where $T_c(n,\Delta)$ is the running time of a distributed $c$-coloring algorithm on networks of size $n$ and maximum degree $\Delta$.
\end{lemma}

The mechanism is deterministic unless the coloring algorithm is randomised. By using the best deterministic $(\Delta+1)$-coloring algorithm we get the following.

\begin{theorem} \label{thm:da-arbitrary-tb}
  There exists a deterministic distributed mechanism in the CONGEST model for stable matching when the providers have common preferences. The mechanism is incentive-compatible for the clients and the algorithm implementing it runs in $O(S\Delta_H + \log^* n)$ rounds.
\end{theorem}

Here, the maximum degree of the conflict graph is bounded by $\Delta_H \leq \Delta_C \Delta_P$.

\begin{proof}
  Using the $(\Delta+1)$-coloring algorithm of Maus and Tonoyan~\cite{maus20linial} that runs in $O(\sqrt{\Delta \log \Delta} + \log^* n)$ rounds we can color $H$ with $\Delta_H + 1$ colors in time $O(\sqrt{\Delta_H \log \Delta_H} + \log^* n)$. Using this in Lemma~\ref{lem:mechanism-main}, the algorithm then runs in $S (\Delta_H + 1)$ rounds, giving a total running time $O(S\Delta_H + \log^* n)$.
\end{proof}

We will now proceed to prove Lemma~\ref{lem:mechanism-main}.

\subsection{Proof of Lemma~\ref{lem:mechanism-main}}

If the providers have common preferences with $S$ preference classes, then the deferred acceptance algorithm runs in $O(S)$ rounds in the CONGEST model.

\begin{lemma} \label{lem:da-common-preferences}
  If providers have common preferences $s$ represented with at most $S$ preference classes and the preferences of each provider form a strict linear order over its adjacent clients, then the distributed deferred acceptance algorithm runs in $2S-1$ rounds in the CONGEST model of distributed computing.
\end{lemma}

\begin{proof}
  This follows from the fact that in each proposal-answer iteration, consisting of 2 communication rounds, all clients with the highest preference class remaining become matched and stop. 
  
  In the first round, the proposals of agents in preference class 1 are always accepted by the providers, as the linear order is strict. They will therefore stop proposing and will not participate in the algorithm further. In the third round, clients in the second preference class either had their proposal accepted in the first round, or have their second proposal accepted in the third round. The latter follows from the fact that all clients of (higher) preference one have already stopped. By induction, by round $2i-1$, for $i=1,2,\dots,S$, all agents with preference class at most $i-1$ have stopped. In the final round no acknowledgement is required, as by the strictness of the preferences the clients that are active before round $2S-1$ are the only such neighbours of their active adjacent providers.
\end{proof}

The score function $s \circ \phi$ constructed in step~\ref{step:coloring} of the algorithm induces a strict linear order for each provider over its adjacent clients, as required in Lemma~\ref{lem:da-common-preferences}. The total number of preference classes is $S c$.

\begin{lemma} \label{lem:linear-order-extension}
  Given a score function $s$ and a $c$-coloring $\phi$, the score function $s \circ \phi$, restricted to the neighborhood of any single provider, forms a strict linear order. The number of preference classes is $S c$.  
\end{lemma}

\begin{proof}
  By construction, the new score function assigns values between 1 and $Sc$. If two clients $v$ and $u$ have $s(v) \neq s(u)$, then by construction they also have $(s \circ \phi)(v) \neq (s \circ \phi)(u)$. Finally, if they have $s(v) = s(u)$ and share a neighbour, then $\{u,v\}$ is in the conflict graph and $\phi(v) \neq \phi(u)$ and therefore $(s \circ \phi)(v) \neq (s \circ \phi)(u)$.
\end{proof}

We also want the algorithm to be a mechanism, i.e.\ revealing the true preferences should be a dominant strategy. To prove this in the special case of common preferences for providers, we use the following technical lemma that states that the match of a client is its highest preference after the matches of clients with a higher preference have been removed.

\begin{lemma} \label{lem:da-common-structure}
  Assume that the providers have common preferences $s$ such that the preferences of each provider form a strict linear order over its adjacent clients. Let $C(i)$ denote the set of clients $v$ with $s(v) = i$, and let $M(i)$ denote the matching obtained by running the deferred acceptance algorithm in the subgraph induced by $\cup_{j=1}^i C(i) \cup P$. For all $i$ and each $v \in C(i)$ it holds that the match of $v$ in $G = (C \cup P, E)$ under the deferred acceptance algorithm is its most preferred match in the subgraph of $G$ induced by $(C \cup P) \setminus M(i-1)$. 
\end{lemma}

\begin{proof}
  We note three facts about the deferred acceptance algorithm. First, each client can only become matched with its $k$th preferred match if each of the $k-1$ more preferred providers 1) have already been proposed to by $v$, and those proposals were rejected, 2) have unmatched $v$ due to a later more preferable proposal, or 3) have announced a match with a client with a better priority (P1). Second, each provider only accepts proposals from more preferred clients and therefore its match can only improve during the algorithm (P2). Third, as shown in the proof of Lemma~\ref{lem:da-common-preferences}, all clients in $C(1), \dots, C(j)$ have stopped after round $2j$ (P3).
  
  The proof is by induction on $i$. For the base case $i=1$, the first proposal of each $v \in C(1)$ is accepted. Therefore $M(1)$ consists of the most preferred match of each client. Now assume the claim holds up to some $j$. Consider $v \in C(j+1)$ and round $2j+1$. If $v$ is matched to some provider $p$, all higher priority providers are matched to a higher priority client by P1 and P2. By P3, $v$ will remain matched to $p$ until the end of the execution. If $v$ is unmatched, let $p'$ be the highest priority provider adjacent to $v$ that has not announced that it is matched to a client in $\cup_{i=1}^j C(i)$. Since $v$ is unmatched, it has not yet proposed to $p'$ (P2). It will now propose to it, and, since higher priority clients have stopped, become matched. If no $p'$ exists, then by (P1) and (P2) all neighbours of $v$ are matched to a more preferred client.
\end{proof}

We can now show Lemma~\ref{lem:mechanism-main}.

\begin{proof}[Proof of Lemma~\ref{lem:mechanism-main}]
  By Lemma~\ref{lem:linear-order-extension}, the score function $s \circ \phi$ gives a partial linear order for the providers with $Sc$ classes such that the order is strict for any provider. By Lemma~\ref{lem:da-common-preferences} the distributed deferred acceptance algorithm then runs in $2Sc-1$ rounds. Constructing the virtual conflict graph requires 1 round of communication. The coloring algorithm, by definition, takes $O(T_c(n, \Delta_H))$ rounds. In total the running time is $O(T_c(n,\Delta)+Sc)$ rounds.
  
  The algorithm computes a stable matching by the combination of Lemmas~\ref{lem:da-common-preferences} and \ref{lem:linear-order-extension}. It remains to argue that the algorithm is incentive-compatible for the clients. The first step, where the conflict graph $H$ is formed, depends on the preferences of the providers, but not on the preferences of the clients. Therefore clients cannot affect the structure of $H$ by lying. In the second step, the coloring algorithm only depends on the structure of $H$ and (possibly) the unique identifiers of the clients. Since we assume that these are not part of the private input, the clients cannot affect the coloring by lying. It is known that the final step, the original deferred acceptance algorithm is incentive compatible for the clients~\cite{dubins81machiavelli,roth82economics}. In the case of common preferences we have a simpler argument for this. By Lemma~\ref{lem:da-common-structure}, the match of each client is its most preferred provider once the matches of higher priority clients have been removed. Conversely, since the match of the higher priority clients is not affected by the preferences of lower priority clients, the agents should always report their true preferences so that they get matched to their most preferred remaining provider.
\end{proof}

\subsection{Coloring in the CONGEST model} \label{ssec:congest-col}

In Section~\ref{sec:def} we assumed that the the communication network has edges between clients that are connected to the same provider (specifically we need edges between clients who are also in the same preference class under $s$). Removing this assumption has two effects: simulating the coloring algorithm in the conflict graph slows down by a factor of 2, as the messages have to be relayed through the providers, and providers form bottle necks since, in the worst case, all connected clients want to send a message to each other connected client.

First, it should be noted that in this case it is impossible to form the explicit conflict graph in constant time. However, since we need $\Omega(\Delta_H)$ time to run the deferred acceptance algorithm, we can spend $O(\Delta_P) = O(\Delta_H)$ rounds to do this.

Most suitable coloring algorithm depends on the maximum degree of the conflict graph. When $\Delta_H = \Omega(\log n)$, we can use a simple variant of Luby's coloring algorithm to color with $2\Delta$ colors. In each round, each uncolored node tries to pick a free color uniformly at random. Since there are $2\Delta$ possible colors, probability to choose a color that is not picked by others is at least 1/2. Providers can be used to verify whether a trial was successful: each client sends its color to each provider and each provider indicates whether a neighbor in the conflict graph chose the same color. If no provider indicates a conflict, the color can be chosen safely. This algororithm runs in $O(\log n)$ rounds with high probability. We note that greedy color reduction does not work in this setting: simulating one round requires sending $O(\Delta_P)$ colors.

For fast $(\Delta+1)$-coloring, we can use the additive-group coloring algorithm of Barenboim and Elkin~\cite{BarenboimEG22} combined with the $O(\Delta^2)$-coloring algorithm of Linial~\cite{Linial1992}. For sufficiently small values of $\Delta_H = O(\log / \log \log n)$ there is no overhead (as again, we must spend $\Omega(\Delta_H)$ rounds to compute the stable matching). For $\Delta_H$ larger than this, there is a slowdown in the algorithm, at least without further optimisation. 

\section{Generalisation to fractional loads and capacities} \label{sec:generalisation}

% !TEX root = sample-acmsmall-submission.tex

In this section we show that our algorithm can be generalised to the case where each client $v$ has some load $\ell(v)$ and each provider $u$ has some capacity $\ell(u)$, these have fractional values, and the goal is to find a stable matching such that for each provider, the total load assigned to it is at most its capacity. Note that this also captures the setting where clients have discrete load that is matched to the providers.

More formally, we call the following problem the \emph{fractional capacitated stable matching} problem. The output is a fractional matching $m \colon E \rightarrow \mathbb{Q}$ that assigns a rational number to each edge. For agent $v$, let $m(v) = \sum_{e: v \in e} m(e)$. We have a capacity constraint $m(v) \leq \ell(v)$ for all agents. We say that agents $v \in C$, $u \in P$: $\{v,u\} \in E$ form a \emph{blocking pair} if they would prefer to assign more load to the matching between them. More formally, if 1) there exist $v' \in N(u)$ such that $m(\{v',u\}) > 0$ and $u >_v u'$, or $m(u) < \ell(u)$ and 2) there exists $u' \in N(v)$ such that $m(\{v,u'\}) > 0$ and $v >_u v'$ or $m(v) < \ell(v)$.

With general preferences, it is not clear how utility should be defined: how does an agent prefer $x$ units of load matched at preference $y$ against $2x$ units of load matched at preference $y+1$. A typical solution is to define a new utility function that fixes some relation between these two. Here we will show that the natural greedy modification of the deferred acceptance algorithm optimises both when one side has common preferences: the clients will maximise both the quality and quantity of their their matching by being truthful.

\subsection{Modified deferred acceptance algorithm}

To account for capacities (and fractional loads), we modify the original deferred acceptance algorithm. We call this simple modification \emph{batch proposals}. First, in each acknowledgement round, all providers send their current loads with the associated preference scores to all of their neighbours. Second, in each proposal round, each client proposes greedily based on the available capacity of its neighbours.

We use the basic algorithm template given in Section~\ref{ssec:basic-alg}, but use the following, modified subroutine for deferred acceptance. For simplicity, we present the non-adaptive version of this algorithm. It has the same worst-case running time as the adaptive version presented in Section~\ref{ssec:basic-alg}, but has slightly better message complexity.

Each client $v$ holds a vector $r(v) = (r_1(v), r_2(v), \dots, r_d(v))$, consisting of the remaining capacities of its neighboring providers and where $d = \deg(v)$. On round 0, these are set to the initial capacities of the neighbours. For a given vector $r(v)$, the greedy batch proposal of $v$ to the $i$th preferred provider is calculated as $p(v,i) = \min \{ r_i(v), l(v) - \sum_{j=1}^{i-1} p(v, j) \}$. That is, go over the providers one by one in preference order and assign the maximum amount of load that is feasible.

On odd rounds $2t-1$ for $t=1,\dots, S$:
\begin{enumerate}[noitemsep]
  \item For a client, if $s(v) = t$, for each neighbor $u_i$ send them the batch proposal $p(v,i)$ and stop with $m(\{v,u_i\}) = p(v,i)$.
  \item If a provider receives a non-zero proposal, accept it.
\end{enumerate}

On even rounds $2t$ for $t=1,\dots, S$:
\begin{enumerate}[noitemsep]
  \item Each provider, if it accepted a non-zero proposal in the previous round, sends its new remaining capacity to all neighbours.
  \item Each client $v$ updates $r(v)$ based on received messages.
\end{enumerate}

Using this subroutine, we show the following lemma.

\begin{lemma} \label{lem:batch-da}
    There exists a distributed mechanism in the CONGEST model for fractional capacitated stable matching when the providers have common preferences. The mechanism is incentive-compatible for the clients and the algorithm implementing it runs in $O(T_c(n,\Delta_H)+Sc)$ rounds.
\end{lemma}

Analogous to Theorem~\ref{thm:da-arbitrary-tb}, applying a graph coloring algorithm yields the following deterministic algorithm.

\begin{theorem} \label{thm:fractional-det-da}
  There exists a deterministic distributed mechanism in the CONGEST model for fractional capacitated stable matching when the providers have common preferences. The mechanism is incentive-compatible for the clients and the algorithm implementing it runs in $O(S\Delta_H + \log^* n)$ rounds.
\end{theorem}

While all fractional capacitated matchings are not are not comparable based on which an agent would prefer, we show that in the case of common preferences, we can define a natural partial order over fractional matchings the algorithm can compute such that agents will always prefer the matching produced when the agent is truthful.
This implies that in the case of common preferences the deferred acceptance algorithm is incentive-compatible for the clients.

For a fixed client $v$, let $m(v,i)$ denote the load assigned by fractional matching $m$ on the edge between $v$ and its $i$th most preferred provider. We define that $m$ is preferable to $m'$ for client $v$ if the following holds: for all $i \in \{1,\dots,\deg(v)\}$, we have that $\sum_{j=1}^i m(v,j) \geq \sum_{j=1}^i m'(v,j)$. That is, assignment $m$ always assigns at least as much load at at least as high preference as $m'$. 

To prove the truthfulness of the modified algorithm, we show that under the modified deferred acceptance algorith, the solution computed under true preferences is preferable to the solution computed under any other reported preferences.

\begin{lemma} \label{lem:batch-da-monotone}
	For each client $v$ and each $i$ in $\{ 1,\dots,\deg(v) \}$, the modified deferred acceptance algorithm computes an assignment such that reporting the true preferences maximises $\sum_{j=1}^i m(v,j)$.
\end{lemma}

This lemma means that by truthful reporting of preferences a client maximises the total matching assigned to any first $i$ most preferred providers. In particular, it means that it also maximises the total weight of its matching.

\begin{proof}
By Lemma \ref{lem:linear-order-extension}, the score function $s \circ \phi$ provides a strict linear order for each provider. Therefore when client $v$ makes its proposal, the adjacent providers receive no other proposals. At the beginning of each round $2t-1$ for $t = 1, \dots,S$, the clients (in preference class $t$) know the available capacities $r_t(v)$ of connected providers. Denote the capacity of the $i$th preferred provider $u$ by $r_t(v,i)$.

By round $2t-1$, all clients in preference classes $i < t$ have finalised their matches. It is important to note that a client cannot influence its preference class $(s\circ \phi)(v)$, determined by the coloring algorithm, nor can it affect the resources assigned to higher priority classes. Further, since the lower priority clients do not replace any of the load assigned in round $t$, the best $v$ can do is to optimise its matching given the current remaining capacities $r_t(v)$.

Now let $>_v$ denote the true preferences of $v$ and let $>'_v$ denote some different preference order. 
Let $m$ and $m'$ denote the assignments computed under $>_v$ and $>'_v$, respectively (with other preferences remaining constant). For each provider in $N(v)$, there are exactly two possibilities: either the provider becomes saturated ($v$ assigns $m(\{v,u_i\}) = r_t(v,i))$, or $v$ assigns its remaining load (and therefore assigns 0 load to remaining providers).

Now consider the matchings $m$ and $m'$. Again, there are two cases: either all providers become saturated, or not. In the first case, the assignments $m$ and $m'$ agree, and the claim holds trivially. In the second case, $m$ saturates some first $k$ most preferred providers, and then assigns some load to the $k+1$th most preferred provider. This directly implies that for all $i \in \{ 1,\dots,k\}$, we have that $\sum_{j=1}^i m(v,i) \geq \sum_{j=1}^i m'(v,i)$. However, since $\sum_{j=1}^{k+1} m(v,i) = \ell(v)$, we also must have $\sum_{j=1}^{k+1} m(v,i) \geq \sum_{j=1}^{k+1} m'(v,i)$.
\end{proof}

With Lemma~\ref{lem:batch-da-monotone}, we can prove Lemma~\ref{lem:batch-da}.

\begin{proof}[Proof of Lemma~\ref{lem:batch-da}] 
 
The score function $s$ and a $c$-coloring $\phi$ work the same way for the fractional version of the algorithm. Therefore, by Lemma \ref{lem:linear-order-extension}, the score function $s \circ \phi$ provides a strict order for the providers with $Sc$ preference classes. 

At each round $2t-1$ for $t \in \{1,\dots Sc \}$, the proposals of clients in preference class $t$ make their proposals and stop. Given the one round of additional preprocessing, the modified deferred acceptance algorithm runs in $2Sc$ rounds. The coloring takes, by definition, $O(T_c(n,\Delta))$ rounds and the total running time is $O(T_c(n,\Delta)+S\cdot c)$.

Since the $s \circ \phi$ are strict, each provider receives at most one proposal in each round and these are always accepted. The resulting matching is stable, as the clients compute the best available matching with respect to existing matching. By Lemma~\ref{lem:batch-da-monotone} the algorithm is also incentive-compatible, as the matching computed given true preferences is always preferred by a client.
\end{proof}

Note that here we assumed that initially all capacities are values that can be sent in one round in the CONGEST model. The operations performed are addition and subtraction, and we assume that the values obtained this way can also be sent in the CONGEST model.

\section{Fair tie-breaking} \label{sec:fairness}

In this section we consider two \emph{fair} ways of breaking ties in the deferred acceptance algorithm. This does not affect the truthfulness of the algorithm, but the goal is to provide an additional guarantee: despite having ties, from the perspective of each client, each way of breaking ties is (almost) equally probable. This implies that the algorithm not biased against any of the clients.

We define a tie-breaking rule as a function $s'$ that assigns a new preference class to each client from some set $S'$ such that each pair of clients that are assigned the same preference class by the original score function $s$ and share an adjacent provider, are assigned a different class. Further, it must respect the original order, i.e.\ if $s(v) < s(u)$, then $s'(v) < s'(u)$. The tie-breaking rule given by a proper coloring of the conflict graphs, as described in Section~\ref{ssec:basic-alg} satisfies these properties. 

Tie-breaking is uniformly random, if each legal tie-breaking is equally likely. A uniformly random coloring $\phi$ of the conflict graph $H$, where all clients are colored from a palette of the same size $c$, would give a uniformly random tie-breaking rule $s \circ \phi$, as described in Section~\ref{ssec:basic-alg}. Computing a uniformly random coloring in the distributed setting is hard, so we will consider two different ways of computing almost random colorings.

\subsection{Sampling an almost random coloring} \label{ssec:sample-coloring}

The first approach is to use an existing sampling algorithm to find an almost random coloring. The state of the art distributed sampling algorithm is due to Carlson et al.\ \cite{carlson23improved}.
\begin{theorem}[\cite{carlson23improved}, Theorem 1]
  For all $\varepsilon > 0$, all $\Delta \geq 2$, all $\delta > 0$, and any $k > (11/6+\varepsilon)\Delta$, for any graph $G = (V,E)$ of maximum degree $\Delta$, a random $k$-coloring within total variation distance $\leq \delta$ from uniform can be generated in $O(\log (n/\delta))$ rounds.
\end{theorem}
 
Fischer and Ghaffari present a simpler algorithm that finds a coloring with slightly more colors, but can also be used (in our context) in the CONGEST model.
\begin{theorem}[\cite{fischer18simple}, Theorem 1] \label{thm:sample-coloring}
  A uniform proper proper $q$-coloring of an $n$-node graph with maximum degree $\Delta$ can be sampled within total variation distance $\delta > 0$ in $O\bigl(\log (\frac{n}{\delta})\bigr)$ rounds, where $q = \alpha\Delta$ for any $\alpha > 2$.
\end{theorem}

Given two distributions $\mu, \nu$ over the same space $\Omega$, the \emph{total variation distance} is defined as $\operatorname{d}(\mu, \nu) = \sum_{\sigma \in \Omega}|\mu(\sigma)-\nu(\sigma)|$. 

The disadvantage of using these algorithms compared to the algorithm we will present in Section~\ref{ssec:fast-coloring-failures} is that almost uniform sampling is slower when maximum degree is low: finding a coloring from a distribution with a total variation distance $\leq \delta$ to a random coloring requires $\Omega\bigl((\log (\frac{n}{\delta})\bigr)$ rounds~\cite{weiming17what}.

Using the algorithm of Fischer and Ghaffari we get the following theorem.
\begin{theorem} \label{thm:fair-almost-uniform}
  There exists a randomised distributed mechanism in the CONGEST model that computes a stable matching with high probability when the providers have common preferences. The mechanism is incentive-compatible for the clients, for any $\delta > 0$, the tie-breaking is within total variation distance $\delta$ from the uniformly random tie-breaking, and the algorithm implementing it runs in $O(\Delta_H S + O(\log(n/\delta))$ rounds.
\end{theorem}

We note that this theorem has an analogous variant for the fractional stable matching as well, as the computation of the fractional matching is independent of how the tie-breaking rule was obtained.

\begin{proof}
  The running time and the incentive-compatibility follow from Lemma~\ref{lem:mechanism-main} and Theorem~\ref{thm:sample-coloring} by using $c = 2\Delta+1$. Similarly, the bound on the total variation distance of the coloring implies directly a corresponding bound on the total variation distance of the tie-breaking rule, as each order is produced by the same number of colorings.
  
  We prove that the algorithm runs in the CONGEST model, even if the matching graph is the communication graph. The algorithm from Theorem~\ref{thm:sample-coloring} works as follows: Initially each node $v$ starts with some color $\phi(v)$. In each round, each node samples a new color $\phi'(v) \in [c]$ uniformly at random. Then it sends its current color $\phi(v)$ and new color $\phi'(v)$ to each of its neighbours. If for no neighbor $u$ it holds that $\phi'(v) \in \{ \phi(u), \phi'(u) \}$, then $v$ sets $\phi'(v)$ as its new color. Otherwise there is a potential conflict, and it keeps its old color. To implement this in the CONGEST model, the providers function as intermediaries: for each neighbor $u$, each provider $v$ checks whether there is a conflict with any of $u$'s neighbours in $H$ through $v$ and reports this back to $u$. If no provider of $u$ reports a conflict, $u$ assumes new color $\phi'(u)$.
\end{proof}

\subsection{Fast random coloring with failures} \label{ssec:fast-coloring-failures}

The second approach is to sample a coloring once and to fix any improperly colored parts arbitrarily. The sampling takes constant time, but the bias for the agents that fail to be properly colored is arbitrarily bad. Additionally the error probability is inversely related to the number of colors, and thus requiring high fairness slows down the matching step. We prove the following lemma.

\begin{lemma} \label{lem:fast-eps-fair-coloring}
  For any $0 < \delta < 1$ there exists a randomised distributed coloring algorithm with the following properties:
  \begin{enumerate}[noitemsep, label=(P\arabic*)]
    \item The running time is $O(\sqrt{\Delta \log \Delta} + \log^* n)$ rounds.
    \item There is a subset $F \subseteq V$ such that each node is in $F$ with probability at least $1-\delta$ (and thus the expected size of $F$ is $(1-\delta)n$) and $F$ is colored uniformly at random with $\lceil \delta^{-1} \rceil \Delta$ colors.
    \item Remaining nodes are properly colored with $\Delta+1$ colors from $\{ \lceil \delta^{-1} \rceil \Delta + 1, \dots, (\lceil \delta^{-1} \rceil+1)\Delta+2 \}$.
  \end{enumerate}
\end{lemma}

Recall that in Section~\ref{sec:ff-da} the algorithm requires that the conflict graph of clients is colored. In what follows, let $H = (V,E)$ denote the conflict graph and let $\Delta$ denote the maximum degree of the conflict graph. As input, the algorithm takes parameter $\delta$: this is the target probability for coloring failure at a single node. Our first coloring algorithm consists of three steps:
\begin{enumerate}[noitemsep]
  \item Each node $v$ picks a color $\phi(v)$ uniformly at random from $\{1, \dots, \lceil \delta^{-1} \rceil \Delta \}$. Here assume that $\lceil \delta^{-1} \rceil \Delta \geq \Delta+1$.
  \item Let $X$ denote the set of nodes with coloring conflicts: if $v \in X$ then there exists $u$ such that $\phi(v) = \phi(u)$. Color nodes in $X$ with $\Delta+1$ colors from the palette $\{\lceil \delta^{-1} \rceil\Delta+1, \dots, (\lceil \delta^{-1} \rceil+1)\Delta+2\}$.
\end{enumerate}

The idea behind this simple algorithm is that all nodes that are either properly colored in step 1 are colored completely at random. The fraction of nodes colored properly in step 1 is at least $(1-\delta)$. The remaining nodes are colored in a way that does not affect the matching computed by nodes colored in step 1.

\begin{proof}[Proof of Lemma~\ref{lem:fast-eps-fair-coloring}]
  The running time of the algorithm is dominated by step 2, as step 1 requires no communication (we assume that the value $\Delta$ or an upper bound for it is common knowledge). Computing a $(\Delta+1)$-coloring of the nodes that fail in step 1 can be computed in $O(\sqrt{\Delta \log \Delta} + \log^* n)$ rounds~\cite{maus20linial}. This establishes P1.
  
  Next, we establish property 2. In step 1, each node $v$ picks a random color from $\{1, \dots, \lceil \delta^{-1} \rceil \Delta \}$. For each neighbor, the probability of a conflict is $1/(\lceil \delta^{-1} \rceil \Delta)$. By a union bound, the total probability of a conflict over at most $\Delta$ neighbors is bounded by $1/\lceil \delta^{-1} \rceil \leq 1/\delta$.
  
  The coloring is uniformly random for nodes in $F = V \setminus X$, as conditioned on having no conflicts, each color is equally likely. This establishes P2.
  
  Finally, P3 follows from the fact that subgraph induced by $X$ has maximum degree $\Delta$ and can therefore be colored with $\Delta+1$ colors.
\end{proof}

Using this algorithm for the template given in Lemma~\ref{lem:mechanism-main} we get the following. The tie-breaking rule guarantees that there is a large set of agents that have fair tie-breaking in the sense that 1) the tie-breaking is completely random among them, and 2) the remaining agents have worse tie-breaking. The trade-off is that the size of the fair set affects the number of colors and thus the running time of the deferred acceptance algorithm.

\begin{theorem} \label{thm:fast-fair-da}
  For any $\delta > 0$, there exists a randomised distributed mechanism in the LOCAL model for stable matching when the providers have common preferences. The mechanism is incentive-compatible for the clients, there exists a subset $F$ of clients of expected size $(1-\delta)$ such that the tie-breaking is uniformly random when restricted to $F$, and $C \setminus F$ has arbitrary tie-breaking at lower preference than $F$. The running time of the algorithm is $O(\delta^{-1} S\Delta_H + \log^* n)$.
\end{theorem}
  
\begin{proof}
  The properties of the algorithm follow from Lemma~\ref{lem:mechanism-main} and Lemma~\ref{lem:fast-eps-fair-coloring}.
\end{proof}

Again, this theorem has an analogous variant for the fractional stable matching as the modification to the deferred acceptance algorithm is unaffected.

\bibliographystyle{ACM-Reference-Format}
\bibliography{bibliography}

\end{document}